\documentclass[11pt,a4paper]{article}
\usepackage{indentfirst,mathrsfs}
\usepackage{amsfonts,amsmath,amssymb,amsthm}
\usepackage{latexsym,amscd}
\usepackage{amsbsy}
\usepackage{color}

\newtheorem{thm}{Theorem}
\newtheorem{lem}{Lemma}

\newtheorem{example}{Example}
\newtheorem{remark}{Remark}

\begin{document}
 \title{Infinite families of constacyclic codes supporting 3-designs and their applications in coding theory}
\author{
Hongsheng Hu\thanks{H. Hu and X. Zeng are with the Key Laboratory of Intelligent Sensing System and Security (Hubei University), Ministry of Education and the Hubei Key Laboratory of Applied Mathematics, Faculty of Mathematics and Statistics, Hubei University, Wuhan 430062, China. 
Emails: 202011104010016@stu.hubu.edu.cn, xiangyongzeng@aliyun.com},
Nian Li\thanks{N. Li is with the Key Laboratory of Intelligent Sensing System and Security (Hubei University), Ministry of Education, the Hubei Provincial Engineering Research Center of Intelligent Connected Vehicle Network Security, School of Cyber Science and Technology, Hubei University, Wuhan 430062, China. 
Email: nian.li@hubu.edu.cn},
Yanan Wu\thanks{Y. Wu is with the School of Mathematical Sciences and Key Laboratory of Pure Mathematics and Combinatorics (LPMC), Nankai University, Tianjin 300100, China. 
Email: yanan.wu@aliyun.com},
Xiangyong Zeng
}
\date{}
\maketitle

\begin{quote}
{\small {\bf Abstract:}  Constacyclic codes over finite fields are of theoretical importance as they are closely related
 to a number of areas of mathematics such as algebra, algebraic geometry, graph theory,
 combinatorial designs and number theory. However,  the study of constacyclic codes in this context remains limited compared to classical cyclic codes. This paper provides two  infinite families of $\lambda$-constacyclic codes over $\mathbb{F}_{q^2}$ that support infinite families of 3-designs, which generalize the results in [IEEE Trans. Inf. Theory 69(4): 2341-2354, 2023]. The parameters and weight distributions are determined completely. Besides, we study their  subfield subcodes and applications on constructing entanglement-assisted quantum error-correcting codes (EAQECCs) and locally recoverable codes (LRCs). It is worthy to mention that two classes of maximal entanglement EAQECCs with a negative or  a high positive net rate are derived. Moreover, two classes of distance-optimal and dimension-optimal LRCs are also obtained.
}

\end{quote}

\vspace{-1\baselineskip}

\section{Introduction}
Constacyclic codes over finite fields provide a natural and powerful generalization of cyclic codes, and hence constitute an important class of linear codes in coding theory. Classical cyclic codes are known to possess elegant algebraic structures, which often lead to efficient decoding algorithms. While $\lambda$-constacyclic codes with $\lambda \ne 1$ have a slightly weaker algebraic structure than cyclic codes, i.e $\lambda = 1$ they still offer structural advantages compared with general linear codes. In fact, their distinct algebraic properties make them special advantages compared to classical cyclic codes.  For example, $\lambda$-constacyclic MDS codes over $\mathbb{F}_q$ with  $\lambda\ne 1$ \cite{krishna1990pseudocyclic} are an interesting alternative to cyclic MDS codes in applications where specific values of $n$ and $k$ (for which a cyclic MDS code does not exist)  are required.  Moreover, as demonstrated in \cite{DDR,SHZ}, for certain choices of $q$, $n$, and $k$, the best constacyclic code over $\mathbb{F}_q$  has a
much better error-correcting capability than the best $[n,k]$ cyclic code over $\mathbb{F}_q$.  In addition, several families of linear codes with specific desirable parameters can be realized through constacyclic codes, but not through cyclic codes \cite{B,FWF,HD,JLX}.

The interplay between coding theory and combinatorial $t$-designs further enhances the importance of this subject. One of the major approaches to the construction of combinatorial $t$-designs is the employment of error-correcting codes.  The reader is referred to \cite{ding2020infinite,ding2020linear,ding2021linear,ding2022designs,tang2021infinite} for more information. Let $\alpha$ be a primitive element of $\mathbb{F}_{q^2}$ with $q=p^m$ and $\beta=\alpha^{q-1}$.  In 2020,  Ding and Tang \cite{ding2020infinite} made a breakthrough by presenting infinite families of near MDS codes  with nonzeros  $\beta$ and $\beta^{2}$ over $\mathbb{F}_{3^m}$ and $\mathbb{F}_{2^m}$, supporting infinite families of 3-designs and 2-designs, respectively. Their work resolved a long-standing open problem on the existence of infinite families of near-MDS codes supporting infinite 2-designs.   In \cite{xiang2022infinite}, Xiang, Tang and Liu proposed a family of cyclic codes with nonzeros  $\beta^3$ and $\beta^{4}$ over $\mathbb{F}_{7^m}$ supporting an
infinite family of 3-designs. Afterwards, Geng and Yang et al. \cite{conj} studied the cyclic codes over $\mathbb{F}_{3^m}$ with nonzeros  $\beta^{4}$ and $\beta^{5}$. They conjectured that this code is a  near MDS code.  In \cite{QWH}, the authors gave an affirmative answer to the conjecture proposed by \cite{conj}. It should be noted that this code also supports 3-designs even though the property did not been studied in \cite{conj,QWH}.  Very recently, Wang, Tang and Ding  investigated an infinite family of cyclic codes over $\mathbb{F}_{q}$ with nonzeros $\beta^{\frac{p^{s}-1}{2}}$ and $\beta^{\frac{p^{s}+1}{2}}$ \cite{wang2023infinite}, where $p$ is an odd prime, $m\geq 2$ is an integer and $s$ is an integer with $1\leq s<m$. They proved that these cyclic codes support 3-designs and thus, their results covered those of  \cite{ding2020infinite,conj,QWH,xiang2022infinite}.

However, despite these advances, the literature on infinite families of cyclic codes supporting 3-designs remains relatively limited. To the best of our knowledge,  only one infinite family of $\lambda$-constacyclic codes supporting 3-designs with $\lambda\neq 1$ \cite[pp. 355]{ding2022designs}, and  two infinite families of negacyclic codes over $\mathbb{F}_{q^2}$ supporting infinite families of 3-designs \cite{wang2023infinite} were proposed so far.  This scarcity of results highlights a significant gap: although constacyclic codes are theoretically important, systematic results on $\lambda$-constacyclic codes with $\lambda \ne 1$ remain rare compared to the abundant references  on cyclic codes. Therefore, one of the objectives of this paper is  to  present infinite families of $\lambda$-constacyclic codes holding an infinite family of 3-designs.  It is worthy to mention that our results generalize those of \cite{wang2023infinite}.

Beyond their intrinsic theoretical significance, constacyclic codes also have important applications in the construction on coding theory, including  locally
repairable codes \cite{CFXF,SZL} and quantum codes \cite{CLZ,ZSL}. Entanglement-assisted quantum error-correcting codes (EAQECCs), a class of important  quantum codes, was firstly introduced by Bowen \cite{Bowen2002} and later developed by Brun et al. \cite{Brun2006}. EAQECCs provide significant  advantages over classical quantum codes by exploiting pre-shared entanglement, allowing them to overcome  the quantum Singleton bound and they can be constructed from arbitrary classical codes without requiring strict self-orthogonality.  Locally Repairable Codes (LRCs) are a class of error-correcting codes that enable the recovery of each encoded symbol by accessing only a small number of other symbols. Key properties such as locality and availability make LRCs particularly advantageous in distributed storage systems. Thus, the other objective of this paper is to study the properties of our codes. In particular, we explore their subfield subcodes and obtain a class of ovoid codes as a byproduct. Besides, we consider the connections between our proposed constacyclic codes and LRCs as well as EAQECCs, and we ultimately derive distance-optimal and dimension-optimal LRCs. Moreover, two families of maximal entanglement EAQECCs with good net rate are obtained.

The remainder of this paper is organized as follows. Section 2 recalls necessary notations and basics. Section 3 studies two families of constacyclic codes over $\mathbb{F}_{q^2}$ and their duals, and presents several infinite families of 3-designs supported by these codes. Section 4 focuses on their subfield subcodes and establishes links to EAQECCs and LRCs. Finally, Section 5 concludes the paper.

\section{Definitions and  preliminaries}
 Let $q$ be a prime power. In this section, we recall some notation about constacyclic codes,  $t$-designs and a class of polynomials, which will be used in subsequent sections.

\subsection{Constacyclic codes over finite fields}
An $[n,k,d]$ code $\mathcal{C}$ over $\mathbb{F}_{q}$ is a $k$-dimensional linear subspace of $\mathbb{F}_q^{n}$ with  length $n$ and minimum Hamming distance $d$. Let $A_{i}$  be the number of codewords $c\in\mathcal{C}$ with Hamming weight $i$. Then the weight enumerator of $\mathcal{C}$ is defined by $1+A_{1}x+A_{2}x^{2}+\cdots+A_{n}x^{n}$ and  the sequence $(1,A_{1},A_{2},\ldots,A_{n})$ is called the  weight distribution of $\mathcal{C}$.
Let $\lambda\in\mathbb{F}_{q}^*$. A linear code $\mathcal{C}$ of length $n$ over $\mathbb{F}_{q}$ is called a $\lambda$-constacyclic code if $(c_{0},c_{1},\ldots,c_{n-1})\in \mathcal{C}$ implies that $(\lambda c_{n-1},c_{0},\ldots,c_{n-2})\in \mathcal{C}$. By definition, 1-constacyclic codes and -1-constacyclic codes are the classical cyclic codes and negacyclic codes, respectively.

Let $\Phi$ be the mapping from $\mathbb{F}_q^n$ to the quotient ring $\mathbb{F}_q[x]/\langle x^{n}-\lambda\rangle$ defined by
\[
\Phi((c_{0},c_{1},\ldots,c_{n-1}))=\sum_{i=0}^{n-1}c_{i}x^i.
\]
It is well known that  $\mathcal{C}\subset \mathbb{F}_q^n$ is  $\lambda$-constacyclic  if and only if $\Phi(\mathcal{C})$ is an ideal of  $\mathbb{F}_q[x]/\langle x^{n}-\lambda\rangle$. We thus can identify $\mathcal{C}$ with $\Phi(\mathcal{C})$.  Since each ideal of the ring $\mathbb{F}_q[x]/\langle x^{n}-\lambda\rangle$  is principal, every $\lambda$-constacyclic code $\mathcal{C}$ can be expressed as $\mathcal{C}=\langle g(x) \rangle$, where $g(x)$ is  monic   with the smallest degree and $g(x) \mid (x^{n}-\lambda)$. The polynomials $g(x)$ and $h(x)=(x^{n}-\lambda)/g(x)$  are called the generator polynomial and  check polynomial  of $\mathcal{C}$, respectively, and the roots of $g(x)$ and $h(x)$ are called zeros and nonzeros of $\mathcal{C}$, respectively.

The algebraic structure of constacyclic codes and their duals can be characterized through their generator and check polynomials as below.
\begin{lem}{\rm (\cite{Dinh2010,krishna1990pseudocyclic}}\label{generator and check poly}) The dual code of an $[n, k]$ $\lambda$-constacyclic code $\mathcal{C}$ generated by $g(x)$ is an $[n, n-k]$ $\lambda^{-1}$-constacyclic code  $\mathcal{C}^\perp$  generated by $\widehat{h}(x) = h_0^{-1}x^kh(x^{-1})$, where $h(x) = (x^n - \lambda)/g(x)$ is the check polynomial of $\mathcal{C}$ and $h_0$ is the coefficient of $x^0$ in $h(x)$.
\end{lem}
Let $\mathbb{Z}_N$ denote the ring of integers modulo $N$ and assume that $\gcd(N, q) = 1$. Let $h$ be an integer with $0 \leq h < N$. The $q$-cyclotomic coset modulo $N$ of $h$ is defined by
$$C_h^{(q, N)} = \{h, hq, hq^2, \ldots, hq^{\ell_h-1}\}\mod N \subseteq \mathbb{Z}_N,$$
where $\ell_h$ is the smallest positive integer such that $h \equiv hq^{\ell_h} \pmod{N}$, and is the size of the $q$-cyclotomic coset $C_h^{(q, N)}$.

The following lemma provides a trace representation for constacyclic codes, which is instrumental in analyzing their weight distributions.
\begin{lem}{\rm (\cite{sun2020class})}\label{trace exp} Let $\lambda \in \mathbb{F}_{q}^*$ with $\text{ord}(\lambda) = r$. Let $n$ be a positive integer such that $\gcd(n, q) = 1$. Define $m = \text{ord}_{rn}(q)$ and let $\gamma \in \mathbb{F}_{q^m}$ be a primitive $rn$-th root of unity such that $\gamma^n = \lambda$. Let $\mathcal{C}$ be a q-ary $\lambda$-constacyclic code of length $n$. Suppose that $\mathcal{C}$ has  $s$ pairwise non-conjugate nonzeros, $\gamma^{i_1}, \ldots, \gamma^{i_s}$, which are $s$ roots of its check polynomial. Then $\mathcal{C}$ has the trace representation $\mathcal{C} = \{\mathbf{c}(a_1, a_2, \ldots, a_s) : a_j \in \mathbb{F}_{q^{m_j}}, 1 \leq j \leq s\}$, where
$$
\mathbf{c}(a_1, a_2, \ldots, a_s) = \Big(\sum\nolimits_{j=1}^s \text{Tr}_{q^{m_j}/q}(a_j\gamma^{-ti_j})\Big)_{t=0}^{n-1},$$
with $m_j = |C_{i_j}^{(q, rn)}|$, $\mathrm{Tr}_{q^a/q^b}(\cdot)$   the trace function from $\mathbb{F}_{q^a}$ to $\mathbb{F}_{q^b}$ and $\mathcal{C}_{i_j}^{(q, rn)}$   the $q$-cyclotomic coset of $i_j$ modulo $rn$.
\end{lem}

\subsection{The $t$-design related to linear codes}

 Let $\kappa$ and $v$ be positive integers such that $1\leq \kappa\leq v$. Let $\mathcal{P}$ be a set of $v$ elements and let $\mathcal{B}$ be a set of $\kappa$-subsets of $\mathcal{P}$. The incidence structure $\mathbb{D}=(\mathcal{P},\mathcal{B})$, where the incidence relation is the set membership, is   called a $t$-$(v,\kappa,\eta)$ design, or simply $t$-design, if every $t$-subset of $\mathcal{P}$ is contained in exactly $\eta$ elements of $\mathcal{B}$. The elements of $\mathcal{P}$ are called points, and those of $\mathcal{B}$ are referred to as blocks. The set $\mathcal{B}$ is called the block set of the design. A $t$-design is said to be simple if $\mathcal{B}$ does not contain any repeated blocks. A $t$-$(v,\kappa,\eta)$ design is called a Steiner system if $t\geq 2$ and $\eta=1$, and is denoted by $S(t,\kappa,v)$.
Let $b$ denote the number of blocks in $\mathcal{B}$. The parameters of a $t$-$(n,\kappa,\eta)$ design satisfy
\begin{equation}\label{t-design}
  \binom{n}{t}\eta=\binom{\kappa}{t}b.
\end{equation}

Linear codes and $t$-designs are closely related. A design may yield many codes and a code
may give many designs.  Let $\mathcal{C}$ be an $[n,k,d]$ code over  $\mathbb{F}_q$. Let the coordinates of a codeword of $\mathcal{C}$ be indexed by $(0,1,\ldots,n-1)$ and define $\mathcal{P}(\mathcal{C})=\{0,1,\ldots,n-1\}$.
Let $\mathcal{B}_{w}(\mathcal{C})$ denote the set of the supports of all codewords with Hamming weight $w$ in $\mathcal{C}$ without repeated blocks. If the incidence structure $(\mathcal{P}(\mathcal{C}),\mathcal{B}_{w}(\mathcal{C}))$ is a $t$-$(n,w,\eta)$ design for some positive integers $t$ and $\eta$, where $1\leq w\leq n$ and $A_{w}\neq 0$, we say that $\mathcal{C}$ supports or holds a $t$-design and the supports of the codewords of weight $w$ in $\mathcal{C}$ hold or support a $t$-design  \cite{ding2022designs}.  The following theorem shows that the pair $(\mathcal{P}(\mathcal{C}),\mathcal{B}_{w}(\mathcal{C}))$ defined by a linear code is a
$t$-design under certain condition.

\begin{thm}\label{A-M}{\rm (\cite{assmus1969new})}
Let \( \mathcal{C} \) be a linear code over $\mathbb{F}_{q}$ with length \( n \) and minimum weight \( d \). Let \( \mathcal{C}^\perp \) with minimum weight \( d^\perp \) denote the dual code of \( \mathcal{C} \). Let \( t \) (\( 1 \leq t < \min\{d, d^\perp\} \)) be an integer such that there are at most \( d^\perp - t \) weights of  \( \mathcal{C} \)  in the range \(\{1, 2, \ldots, n - t\}\). Then

{\rm  (1)}  \( (P(\mathcal{C}), B_k(\mathcal{C})) \) is a simple \( t \)-design provided that \( A_k \neq 0 \) and \( d \leq k \leq \omega \), where \( \omega \) is defined to be the largest integer satisfying \( \omega  \leq n \) and
    \[
    \omega  - \left\lfloor \frac{\omega  + q - 2}{q - 1} \right\rfloor < d.
    \]

{\rm (2)}  \( (P(\mathcal{C}^\perp), B_k(\mathcal{C}^\perp)) \) is a simple \( t \)-design provided that \( A_k^\perp \neq 0 \) and \( d^\perp \leq k \leq \omega ^\perp \), where \( \omega ^\perp \) is defined to be the largest integer satisfying \( w^\perp \leq n \) and
    \[
    \omega ^\perp - \left\lfloor \frac{\omega ^\perp + q - 2}{q - 1} \right\rfloor < d^\perp.
    \]

\end{thm}
\subsection {The solutions of a class of  polynomials}
To determine the possible weights of codewords of our codes, we need to study the number of solutions of certain  equations over finite fields. The following result is due to Bluher.

\begin{lem}{\rm (\cite{bluher})}\label{rational roots}
Let \(m\), \(k\) be two integers and \(\gcd(m,k)=d\). Let \(q=p^{m}\), where \(p\) is a prime number. Then the polynomial \(f(x)=x^{p^{k}+1}+ax+b\) has exactly \(0\), \(1\), \(2\) or \(p^{d}+1\) rational roots in \(\mathbb{F}_{q}\) when \(a\), \(b\) run through \(\mathbb{F}_{q}^{*}\).
\end{lem}

Let \(U_{q+1}\) be   the set of all $(q+1)$-th roots of unity  of $\mathbb{F}_{q^2}$. Consider the equation
\begin{equation}
bx^{p^k+1}+ax^{p^k}+a^{q}x+b^{q}=0.
\label{eq:main}
\end{equation}
We then have the following lemma.

\begin{lem}\label{root}
Let \(m\), \(k\) be two positive integers and \(q=p^{m}\), where \(p\) is any prime number. Then (\ref{eq:main}) has 0, 1, 2 or \(p^{\gcd(k,m)}+1\) solutions in \(U_{q+1}\) for any \(a,b\in\mathbb{F}_{q^{2}}\) and \((a,b)\neq(0,0)\).
\end{lem}
\begin{proof}
 Let $\overline{\beta}$ be the conjugate of $\beta$, namely, $\overline{\beta} = \beta{q}$.
Note that \(U_{q+1}=\{1\}\cup\{\frac{y+\beta}{y+\overline{\beta}}:y\in\mathbb{F}_{q}\}\) for any \(\beta\notin\mathbb{F}_{q}\) and \(x=1\) is a solution of (\ref{eq:main}) if and only if \(a+a^{q}+b+b^{q}=0\). Therefore, we will consider (\ref{eq:main}) from the following two cases.

Case I: If \(a+a^{q}+b+b^{q}\neq 0\), then \(x=1\) is not the solution of (\ref{eq:main}). Substituting \(x\) with \(\frac{y+\beta}{y+\overline{\beta}}\), we then have
\[
b\left(\frac{y+\beta}{y+\overline{\beta}}\right)^{p^{k}+1}+a\left(\frac{y+\beta}{y+\overline{\beta}}\right)^{p^{k}}+\overline{a}\frac{y+\beta}{y+\overline{\beta}}+\overline{b}=0.
\]
Further, it can reduced to
\begin{equation}
Ay^{p^{k}+1}+By^{p^{k}}+Cy+D=0,
\label{eq:reduced}
\end{equation}
where\begin{align*}
A &= a+a^{q}+b+b^{q}, \\
B &= a\overline{\beta}+\overline{a}\beta+b\beta+\overline{b}\,\overline{\beta}, \\
C &= a\beta^{p^{k}}+\overline{a}\overline{\beta}^{p^{k}}+b\beta^{p^{k}}+\overline{b}\,\overline{\beta}^{p^{k}}, \\
D &= a\beta^{p^{k}}\overline{\beta}+\overline{a}\beta\overline{\beta}^{p^{k}}+b\beta^{p^{k}+1}+\overline{b}\,\overline{\beta}^{p^{k}+1}.
\end{align*}
Thus, it is equivalent to determine the number of solutions of (\ref{eq:reduced}) in \(\mathbb{F}_{q}\). Note that \(A\), \(B\), \(C\), \(D\in\mathbb{F}_{q}\). Therefore, in this case, (\ref{eq:reduced}) has 0, 1, 2 or \(p^{\gcd(k,m)}+1\) solutions in \(\mathbb{F}_{q}\) when \(a\), \(b\) runs over \(\mathbb{F}_{q^{2}}\) and \((a,b)\neq(0,0)\) according to Lemma \ref{rational roots}.

Case II: If \(a+a^{q}+b+b^{q}=0\), then \(x=1\) is one of the solutions of (\ref{eq:main}). Next, we need to consider the number of solutions in \(U_{q+1}\backslash\{1\}\). Similar treatment as in Case I, it is sufficient to solve the following equation
\[
By^{p^{k}}+Cy+D=0.
\]
Firstly, we claim that \(B\), \(C\) and \(D\) cannot be 0 simultaneously. Suppose \(B=C=D=0\), then we have
\[
\begin{cases}
a+a^{q}+b+b^{q} = 0, \\
a\overline{\beta}+\overline{a}\beta+b\beta+\overline{b}\,\overline{\beta} = 0, \\
a\beta^{p^{k}}+\overline{a}\overline{\beta}^{p^{k}}+b\beta^{p^{k}}+\overline{b}\,\overline{\beta}^{p^{k}} = 0, \\
a\beta^{p^{k}}\overline{\beta}+\overline{a}\beta\overline{\beta}^{p^{k}}+b\beta^{p^{k}+1}+\overline{b}\,\overline{\beta}^{p^{k}+1} = 0.
\end{cases}
\]
From the first three equations, one can obtain that \(a=-b\in\mathbb{F}_{q}\). Then by the last equality and \((a,b)\neq(0,0)\) , we have \((\overline{\beta}-\beta)^{p^{k}+1}=0\) which implies that \(\beta\in\mathbb{F}_{q}\), a contradiction. Besides, since \(\gcd(p^{k}-1,p^{m}-1)=p^{\gcd(k,m)}-1\), then the linearized polynomial \(ry^{p^{k}}+sy\) has either 1 or \(p^{\gcd(k,m)}\) zeros in \(\mathbb{F}_{q}\) for any \(r,s\in\mathbb{F}_{q}\) and \((r,s)\neq(0,0)\). It means that \(By^{p^{k}}+Cy+D=0\) has 0, 1 or \(p^{\gcd(k,m)}\) solutions. Hence, one can conclude that the number of solutions of (\ref{eq:main}) takes value from \(\{1,\,2,\,p^{\gcd(k,m)}+1\}\) in this case. This completes the proof.
\end{proof}

\begin{remark}
\label{remark:comparison} Although the solutions of equation  \eqref{eq:main} also was analyzed in \cite[ Lemma 8]{wang2023infinite}, the result was stated only for odd $q$. It should be noted that  the proof technique in \cite[ Lemma 8]{wang2023infinite} does not explicitly rely on the parity of $q$ and could potentially be extended to even characteristic. Nevertheless, we still provide an alternative method to prove the solutions of  \eqref{eq:main}, in which we remove the restriction on $q$.
\end{remark}

\begin{remark}\label{remark:conj} In  \cite{conj}, the authors presented a conjecture as follows. Let  $m$ be an  odd integer and  $q=3^m$. Then the dual code $\mathcal{C}_{(q, q+1,3,4)}^{\perp}$ of  $\mathcal{C}_{(q, q+1,3,4)}$ is an almost MDS code with parameters $[q+1,4, q-3]$.
As stated in  \cite{conj}, to solve the above conjecture, it is sufficient to determine that the number of solutions of the equation $b u^{10}+a u^{9}+a^{q} u+b^{q}=0$ in $U_{q+1}$ is no larger than 4 and there exists at least one pair of $(a, b)$ such that the above equation has exactly 4 solutions, where $a, b \in \mathbb{F}_{q^2},\;(a, b) \neq(0,0)$ and $U_{q+1}=\left\{x \in \mathbb{F}_{q^2}: x^{q+1}=1\right\}$. It is clear that this equation is a special case of  \eqref{eq:main} by putting $p=3$ and $k=2$. Lemma \ref{root} has shown that it has $0,1,2$ or $4$ roots in this case. Further, let $a=w^{\frac{q+1}{2}}$ and $b=0$, where $w$ is the primitive element of $\mathbb{F}_{q^2}$. One has immediately obtain that it has 4 solutions. Therefore, our result in Lemma \ref{root} gives an affirmative answer to the conjecture proposed by \cite{conj}.
\end{remark}

\section{Two infinite families of constacyclic codes}\label{main:result}
Throughout this section, we always assume that  \( p \) is any  prime number and \( m \) is a positive integer.  Let \( q = p^m \), \( n = q^2 + 1 \) and $\lambda \in \mathbb{F}_{q^2}^*$ with ${\rm ord}(\lambda)=r$. Let \( \beta \) be a  primitive element of \( \mathbb{F}_{q^4} \) and put \( \delta = \beta^{(q^2 - 1)/r} \). Then it can be seen that \( \delta^n = \lambda \) and \( \delta \) is a primitive \( rn \)-th root of unity of \( \mathbb{F}_{q^4}^* \).

In this section, by giving the restriction on $r$, we present two infinite families of $\lambda$-constacyclic codes of length \( n \) over \(\mathbb{F}_{q^2} \) and study their dual codes.  It shows that these codes and their duals support 3-designs under certain conditions.

%

\subsection{The first family of constacyclic codes supporting 3-designs}\label{first class}
Following the above notation, in this subsection, we always assume that $r={\rm ord}(\lambda)$ satisfies the  two conditions as below:

 (1) $r \mid (q+1)$ and

(2) $\nu_{2}(r)= \nu_{2}(q+1)$,

\noindent where $\nu_2(\cdot)$ is the 2-adic order function.

For each $i$ with $0\leq i < rn$, let $g_{i}(x)$ denote the minimal polynomial of $\delta^{i}$ over $\mathbb{F}_{q^{2}}$. Then it is easy to check that the minimal polynomial of $\delta$ and $\delta^{q^{2}+q+1}$ are
$$
g_{1}(x)=  (x-\delta)(x-\delta^{q^{2}}) \quad {\rm and} \quad
g_{q^{2}+q+1}(x)=  (x-\delta^{q^{2}+q+1})(x-\delta^{q^{4}+q^{3}+q^{2}}),
$$ respectively.
Let $\mathcal{C}(1,q^{2}+q+1)$ denote the constacyclic code of length \( n = q^2 + 1 \) over \( \mathbb{F}_{q^2} \) with the check polynomial \( g_{1}(x)g_{q^{2}+q+1}(x) \). In what follows, we will focus on the parameters of \( \mathcal{C}(1, q^2 + q + 1) \) and its dual, and study their relations with $t$-designs.

Let $U_{s(q^2+1)}$  be the set of all $s(q^2+1)$-th roots of unity in $\mathbb{F}_{q^4}^*$, where $s\mid (q^2-1)$. Define  $f(x)=x^{q+1}$ limited on $U_{r(q^{2}+1)}$,  and the preimage set of $f(x)$ on $y_0$ is defined as $f^{-1}(y_0)=\{x\in U_{r(q^2+1)}\,\mid \,f(x)=y_0\}$, where $y_0\in U_{r(q^{2}+1)}$. For any $0\leq j \leq r-1$, denote by
\begin{eqnarray}\label{T}T=\{\delta^{-i} \mid 0\leq i\leq q^{2}\} \quad{\rm and} \quad \lambda^{j}T=\{\lambda^{j}\delta^{-i} \mid 0\leq i\leq q^{2}\}.\end{eqnarray}
Then we have the following lemma, which is useful to give the main result in this subsection.

\begin{lem}\label{q+1 relation}
Let   notation be  the same as above and $y_{0} \in U_{q^{2}+1}$. Then $f(x)$ is a surjective function from $U_{r(q^{2}+1)}$ to $U_{q^{2}+1}$. Moreover, if $f(x_0)=y_0$ for some $x_0\in U_{r(q^{2}+1)}$, then
  $f^{-1}(y_{0})=\{ \lambda^j x_0\mid 0\leq j \leq r-1 \}$ and $\mid f^{-1}(y_{0})\cap \lambda^{j}T \mid = 1$ for any $0\leq j \leq r-1$.
\end{lem}
\begin{proof}
  Since $r \mid (q+1)$ and $\nu_{2}(r)= \nu_{2}(q+1)$, we have $\gcd(r(q^2+1),q+1)=r$. It yields that  $f(x)$ is a surjective function from $U_{r(q^{2}+1)}$ to $U_{q^{2}+1}$ and $x^{\frac{q+1}{r}}$ is a permutation polynomial on $U_{q^{2}+1}$.
   Further, one obtains that  for each  $y \in U_{q^{2}+1}$, $f(x)=y$ has $r$ roots in $U_{r(q^{2}+1)}$. Since ${\rm ord}(\lambda)=r$ and $r \mid (q+1)$, one can immediately obtain that $f(\lambda^j x_0)=f(x_0)=y_0$. Note that  $\lambda^j x_0, 0\leq j \leq r-1$ are pairwise distinct elements. Therefore,  $f^{-1}(y_{0})=\{ \lambda^j x_0\mid 0\leq j \leq r-1 \}$.

   From \eqref{T}, it can be seen that $| \lambda^j T|=q^2+1$ where $0\leq j\leq r-1$, and for any $a\in \lambda^j T$,  one has $a\in U_{r(q^2+1)}$. Assume that there exist $0\leq j_1, j_2\leq r-1$ and $0\leq i_1,i_2\leq q^2$  with $i_1\ne i_2$ and $j_1\ne j_2$ such that $\lambda^{j_1}\delta^{-i_1}=\lambda^{j_2}\delta^{-i_2}$. Then it yields $\lambda^{j_1-j_2}=\delta^{i_1-i_2}$, which indicates that $\delta^{r(i_1-i_2)}=1$ due to ${\rm ord}(\lambda)=r$. It is impossible due to ${\rm ord}(\delta)=r(q^2+1)$ and $0\leq i_1,i_2\leq q^2$. Therefore, one has $\lambda^{j_1} T\cap \lambda^{j_2} T=\emptyset$ for any $0\leq j_1, j_2\leq r-1$ and  $j_1\ne j_2$. Further, one can conclude that $\cup_{j=0}^{r-1}\lambda^{j}T=U_{r(q^2+1)}$,
    which implies   $f^{-1}(y_{0}) \subset \cup_{j=0}^{r-1}\lambda^{j}T$.
  On the other hand, suppose that there is some $0\leq j_0\leq r-1$ satisfies $\mid f^{-1}(y_{0})\cap \lambda^{j}T \mid >1$. Then there must be  $ 0\le j_{1} < j_{2}\leq r-1$ such that $\lambda^{j_{1}}x_{0}, \lambda^{j_{2}}x_{0}\in \lambda^{j_{0}}T$, i.e.,  $\lambda^{j_{1}}x_{0}=\lambda^{j_{0}}\delta^{-i_{1}}$ and $ \lambda^{j_{2}}x_{0}=\lambda^{j_{0}}\delta^{-i_{2}}$
    for some $0\leq i_{1}, i_{2}\leq q^2$ and $ i_{1}\neq i_{2}$. This gives that $x_{0}=\lambda^{j_{0}-j_{1}}\delta^{-i_{1}}=\lambda^{j_{0}-j_{2}}\delta^{-i_{2}}$, which is impossible again by  ${\rm ord}(\delta)=r(q^2+1)$. Therefore, $\mid f^{-1}(y_{0})\cap \lambda^{j}T \mid \leq1$ for any  $0\leq j_0\leq r-1$. Then the result follows from  $f^{-1}(y_{0}) \subset \cup_{j=0}^{r-1}\lambda^{j}T$ and $|f^{-1}(y_{0}) |=r$. This completes the proof.
\end{proof}
By Lemma \ref{q+1 relation}, we now present the first important conclusion of this paper.

\begin{thm}\label{theorem1} Let notation be the same as above and $q > 2$.
Then the code $\mathcal{C}(1,q^{2}+q+1)$ over $\mathbb{F}_{q^{2}}$ has parameters $[q^{2}+1,4,q^{2}-q]$ and weight enumerator
\[
1+(q^{5}-q)z^{q^{2}-q}+\frac{(q^{4}-1)(q-1)q^{3}}{2}(z^{q^{2}-1}+z^{q^{2}+1})+(q^{7}-q^{5}+q^{4}-q^{3}+q-1)z^{q^{2}}.
\]
Furthermore, the nonzero minimum weight codewords of the code $\mathcal{C}(1,q^{2}+q+1)$ support a $3$-$(q^{2}+1,q^{2}-q,(q^{2}-q-1)(q-2))$ design.
The code $\mathcal{C}(1,q^{2}+q+1)^{\perp}$ over $\mathbb{F}_{q^{2}}$ has parameters $[q^{2}+1,q^{2}-3,4]$ and the nonzero minimum weight codewords of $\mathcal{C}(1,q^{2}+q+1)^{\perp}$ support a $3$-$(q^{2}+1,4,q-2)$ design.
\end{thm}

\begin{proof}To completes the proof, we give the following claims.

{\bf Claim 1.}  The possible nonzero weights of  $\mathcal{C}(1,q^{2}+q+1)$ take values from the set
$$\{q^2+1,q^2,q^2-1,q^2-q\}.$$

 Recall that $\delta\in \mathbb{F}_{q^{4}}$ is a primitive $r(q^{2}+1)$-th root of unity and
\[
g_{1}(x)= (x-\delta)(x-\delta^{q^{2}}), \quad g_{q^{2}+q+1}(x)= (x-\delta^{q^{2}+q+1})(x-\delta^{q^{4}+q^{3}+q^{2}}),
\]
which are two distinct divisors of $x^{n}-\lambda$. Thus,
\[
\deg(g_{1}(x)g_{q^{2}+q+1}(x))=4,
\]
which induces that the dimension of the code $\mathcal{C}(1,q^{2}+q+1)$ is $4$.
We now analyze  the possible nonzero weights of the code $\mathcal{C}(1,q^{2}+q+1)$. According to Lemma \ref{trace exp}, the trace expression of $\mathcal{C}(1,q^{2}+q+1)$ is given by
\[
\mathcal{C}(1,q^{2}+q+1)=\{\mathbf{c}(a,b): a,b\in  \mathbb{F}_{q^{4}}\},
\]
where
\[
\mathbf{c}(a,b)=\left(\operatorname{Tr}_{q^{4}/q^{2}}\left(a\delta^{-i}+b\delta^{-(q^{2}+q+1)i}\right)\right)_{i=0}^{q^{2}}.
\]
Put $x=\delta^{-i}.$ Then $x \in  T\subset U_{r(q^{2}+1)}$ and
\[
\operatorname{Tr}_{q^{4}/q^{2}}\left(ax+bx^{q^{2}+q+1}\right)=ax+a^{q^{2}}x^{q^{2}}+bx^{q^{2}+q+1}+b^{q^{2}}x^{q^{4}+q^{3}+q^{2}},
\]
where $T$ is defined as \eqref{T}. For any $(a,b)\in\mathbb{F}_{q^4}^2\backslash\{(0,0)\}$, denote by
$$N(a,b)=\mid \{ x \in T \mid ax+a^{q^{2}}x^{q^{2}}+bx^{q^{2}+q+1}+b^{q^{2}}x^{q^{4}+q^{3}+q^{2}}=0\} \mid.$$
Then it can be seen from  the  definition that the  weights of  $\mathbf{c}(a,b)$ is $wt(\mathbf{c}(a,b))=n-N(a,b)$.
 Let $d$ be the minimum Hamming distance of  $\mathcal{C}(1,q^{2}+q+1)$. Then
$$d=n-\max\{N(a,b): a,b \in \mathbb{F}_{q^4}\,{\rm and}\,(a,b)\ne (0,0)\}.$$
Observe that  $q^4+q^3+q^2+q=q(q+1)(q^2+1)\equiv 0 \, ({\rm mod} \, r(q^2+1))$. Thus, we have
\begin{equation*}
  \begin{aligned}
  & ax+a^{q^{2}}x^{q^{2}}+bx^{q^{2}+q+1}+b^{q^{2}}x^{q^{4}+q^{3}+q^{2}}
\\=&  ax+a^{q^{2}}x^{q^{2}}+bx^{q^{2}+q+1}+b^{q^{2}}x^{-q}
\\=& x^{-q}(ax^{q+1}+a^{q^{2}}x^{q^{2}+q}+bx^{q^{2}+2q+1}+b^{q^{2}})
\\=& x^{-q}(ay_{0}+a^{q^{2}}y^{q}_{0}+by^{q+1}_{0}+b^{q^{2}}),
  \end{aligned}
\end{equation*}
where $y_{0}=x^{q+1}$.
When $y_{0}$ runs over $U_{q^2+1}$, the possible number of solutions of the equation
$$ay_{0}+a^{q^{2}}y^{q}_{0}+by^{q+1}_{0}+b^{q^{2}}=0$$
is $0,1,2$ or $q+1$ according to Lemma \ref{root}.
Further,  by Lemma \ref{q+1 relation}, we have
$$\mid f^{-1}(y_{0})\cap T \mid= 1,$$  where $f(x)=x^{q+1}$ limited on $U_{r(q^{2}+1)}$. Since $T\subset U_{r(q^2+1)}$, it gives that  $N(a,b)\in\{0,1,2,q+1\}$.
Therefore,  the set of possible nonzero weights of $\mathcal{C}(1,q^{2}+q+1)$  is
$$\{q^2+1,q^2,q^2-1,q^2-q\}.$$

{\bf Claim 2.}  The dimension of the code  $\mathcal{C}(1,q^{2}+q+1)^{\perp}$ is $q^2-3$ and the minimum Hamming distance of $\mathcal{C}(1,q^{2}+q+1)^{\perp}$ is 4 (resp. 5) when $q>2$.

 Let $d^{\perp}$ be the minimum Hamming distance of  $\mathcal{C}(1,q^{2}+q+1)^{\perp}$. It can be checked that $d^{\perp}\geq 2$, and $d\leq q^2-2$, $d^{\perp}\leq 5$  follows by the Singleton bound.

If $d^{\perp}=5$, then  $\mathcal{C}(1,q^{2}+q+1)^{\perp}$ is an MDS code. This means that  $\mathcal{C}(1,q^{2}+q+1)$  has parameters $[q^2+1,4,q^2-2]$. Note that $q^2-1>q^2-2$, $q^2-q\geq 2$ and the equality holds only if $q=2$.  Therefore, we can conclude that $d=q^2-2, d^{\perp}=5$ when $q=2$, which indicates that  $\mathcal{C}(1,q^{2}+q+1)$and  $\mathcal{C}(1,q^{2}+q+1)^{\perp}$ are MDS codes for this case. Otherwise, when $q>2$, we have  $d=q^2-q$ and $2\leq d^{\perp}\leq 4$. In the following, we prove that $d^{\perp}\neq 2$ and $d^{\perp}\neq 3$ for the case $q>2$.

Suppose $d^{\perp} = 2$. Then there exist $a_1 \in \mathbb{F}_{q^{2}}^{*}$ and $1 \leq i \leq q^{2}$ such that
\[
\begin{cases}
1 + a_1 \delta^{-i} = 0, \\
1 + a_1 \delta^{-i(q^{2}+q+1)} = 0.
\end{cases}
\]
Put $x = \delta^{-i }$. Then $x\in T \setminus \{1\}$, where $T$ is defined as \eqref{T}. Moreover, one has
\begin{eqnarray}\label{d=2}
\begin{vmatrix}
1 & x \\
1 & x^{q^{2}+q+1}
\end{vmatrix}
= x(x^{q^{2}+q} - 1) = 0,
\end{eqnarray}
which induces that $x^{q+1}=1$ due to $x\ne 0$.  By Lemma \ref{q+1 relation}, since  $|f^{-1}(1) \cap T| = 1$ and $x\ne 1$, one can derive that  \eqref{d=2} has no solution in  $T \setminus \{1\}$. This gives that $d^{\perp} \ne 2$.

Suppose $d^{\perp} = 3$. Then there exist  $i, j$ with $1 \leq i \neq j \leq q^{2}$, and $a_1, a_2 \in \mathbb{F}_{q^{2}}^{*}$ such that
\begin{eqnarray}\label{d=3-1}
\begin{cases}
1 + a_1 \delta^{-i} + a_2 \delta^{-j} = 0, \\
1 + a_1 \delta^{-i(q^{2}+q+1)} + a_2 \delta^{-j(q^{2}+q+1)} = 0.
\end{cases}
\end{eqnarray}
Raising the first equation of \eqref{d=3-1} to the $q^{2}$-th power gives
\begin{eqnarray}\label{d=3-2}
1 + a_1 \delta^{-iq^{2}} + a_2 \delta^{-jq^{2}} = 0.
\end{eqnarray}
Put $x = \delta^{-i}$, $y = \delta^{-j}$. Then $x \neq y \in T \setminus \{1\}$,  and \eqref{d=3-1} and  \eqref{d=3-2} yield that
\begin{eqnarray}\label{d=3-3}
&&\begin{vmatrix}
1 & x & y \\
1 & x^{q^{2}+q+1} & y^{q^{2}+q+1} \\
1 & x^{q^{2}} & y^{q^{2}}
\end{vmatrix}\nonumber \\
= &&xy \left((x^{(q+1)q} - 1)(y^{(q+1)(q-1)} - 1) - (x^{(q+1)(q-1)} - 1)(y^{(q+1)q} - 1) \right) \nonumber \\=&& 0.
\end{eqnarray}
Let $x_0 = x^{q+1}$, $y_0 = y^{q+1}$. By Lemma \ref{q+1 relation}, one has $x_0 \neq y_0\in U_{q^{2}+1} \setminus \{1\}$. Therefore, \eqref{d=3-3} reduces to
\begin{eqnarray}\label{d=3-4}
(x_0^q - 1)(y_0^{q-1} - 1) = (x_0^{q-1} - 1)(y_0^q - 1).
\end{eqnarray}
Note that $y_0^{q-1}\ne 1$ and $x_0^{q-1}\ne 1$. Without loss of generality, assume that $y_0^{q-1}=1$. Then  $y_0=-1$ due to $\gcd(q-1,q^2+1)=2$ and $y_0\neq 1$. It follows from \eqref{d=3-4} that $x_0=-1$ again by $\gcd(q-1,q^2+1)=2$ and $x_0\neq 1$. This contradicts to $x_0\ne y_0$. Raising both sides of \eqref{d=3-4} to the $q^{2}$-th power yields that
\[
\left( \frac{1}{x_0^q} - 1 \right)\left( \frac{1}{y_0^{q-1}} - 1 \right) = \left( \frac{1}{x_0^{q-1}} - 1 \right)\left( \frac{1}{y_0^q} - 1 \right),
\]
which is equivalent to
\begin{equation*}
\frac{y_0}{x_0}(x_0^q - 1)(y_0^{q-1} - 1) = (x_0^{q-1} - 1)(y_0^q - 1).
\end{equation*}
The above equation together with \eqref{d=3-4}  gives that $\frac{y_0}{x_0}=1$, contradicts to $x_0\ne y_0$.
Hence, $d^{\perp} \neq 3$,  and  $d^{\perp} = 4$. This completes the proof of Claim 2.

Next we completes the proof with the help of {\bf Claim 1} and {\bf Claim 2}.
{\bf Claim 1} shows that the possible nonzero weights of the codewords in $\mathcal{C}(1,q^{2}+q+1)$ are $q^{2}+1$, $q^{2}$, $q^{2}-1$ and $q^{2}-q$. Denote $w_{0}=q^{2}+1$, $w_{1}=q^{2}$, $w_{2}=q^{2}-1$ and $w_{3}=q^{2}-q$. Let $A_{w_{i}}$ denote the number of the codewords with weight $w_{i}$ in $\mathcal{C}(1,q^{2}+q+1)$, where $0\leq i\leq 3$. Since  the minimum Hamming distance of $\mathcal{C}(1,q^{2}+q+1)^{\perp}$ is $4$ by {\bf Claim 2}, we then have
\[
\begin{cases}
\sum_{i=0}^{3}A_{w_{i}}=q^{8}-1,\\
\sum_{i=0}^{3}w_{i}A_{w_{i}}=q^{6}(q^{4}-1),\\
\sum_{i=0}^{3}w_{i}^{2}A_{w_{i}}=q^{8}(q^{4}-1),\\
\sum_{i=0}^{3}w_{i}^{3}A_{w_{i}}=q^{16}+2q^{14}-2q^{12}-2q^{10}+q^{4}.
\end{cases}
\]
based on the first four pless power moments \cite{huffman2003fundamentals}.  Solving this system of equations, we obtain
\[
\begin{cases}
A_{w_{3}}=q^{5}-q,\\
A_{w_{0}}=A_{w_{2}}=(q^{4}-1)(q-1)q^{3}/2,\\
A_{w_{1}}=q^{7}-q^{5}+q^{4}-q^{3}+q-1.
\end{cases}
\]
Moreover,
it follows from the Assmus-Mattson Theorem that the minimum weight codewords in $\mathcal{C}(1,q^{2}+q+1)$ support a $3-(q^{2}+1,q^{2}-q,\eta)$ design. Note that the number of the supports of the minimum weight codewords in $\mathcal{C}(1,q^{2}+q+1)$ is
\[
b=\frac{A_{q^{2}-q}}{q^{2}-1}=\frac{q^{5}-q}{q^{2}-1}=q^{3}+q.
\]
Then from \eqref{t-design} we deduce
\[
\eta=(q^{2}-q-1)(q-2).
\]
On the other hand, with the fifth Pless power moment, one can obtain the number of codewords with weight $4$ in $\mathcal{C}(1,q^{2}+q+1)^{\perp}$, which is given by
\[
A_{4}^{\perp}=\frac{q^{2}(q-2)(q^{2}+1)(q^{2}-1)^{2}}{24}.
\]
Again by the Assmus-Mattson Theorem,  the codewords of weight $4$ in $\mathcal{C}(1,q^{2}+q+1)^{\perp}$ support a $3$-$(q^{2}+1,4,q-2)$ design.

This completes the proof.
\end{proof}

\begin{remark}
In \cite[Theorem 22]{wang2023infinite}, an infinite family of negacyclic codes over $\mathbb{F}_{q^2}$ with parameters $[q^2+1, 4, q^2 - q]$ was constructed under the condition $q \equiv 1 \pmod{4}$, which supports a $3$-$(q^2+1, q^2 - q, (q^2 - q - 1)(q - 2))$ design. Our Theorem  \ref{theorem1} presents a more general family of $\lambda$-constacyclic codes over $\mathbb{F}_{q^2}$ with the same parameters and combinatorial properties, but without any restriction on $q$. Specifically, we require only that $\mathrm{ord}(\lambda) = r$ satisfies $r \mid (q+1)$ and $\nu_2(r) = \nu_2(q+1)$, which includes the case $r = 2$ (i.e., negacyclic codes) when $q \equiv 1 \pmod{4}$. Therefore, our result totally extends that of \cite[Theorem 22]{wang2023infinite}.
\end{remark}

\begin{example}
  Let $q=32$, $r=11$. Then $\mathcal{C}(1,1057)$ has parameters $[1025,4,992]$ and weight enumerator
\[
1+33554400x^{992}+532575436800x^{1023}+34327199775x^{1024}+532575436800x^{1025}.
\]
The dual code $\mathcal{C}(1,1057)^{\perp}$ has parameters $[1025,1021,4]$. The codewords of weight $992$ in $\mathcal{C}(1,1057)$ support a $3$-$(1025,992,29730)$ design, and the codewords of weight $4$ in $\mathcal{C}(1,1057)^{\perp}$ support a $3$-$(1025,4,30)$ design, which is consistent with Theorem \ref{theorem1}.
\end{example}

\begin{example}
  Let $q=29$, $r=6$. Then $\mathcal{C}(1,871)$ has parameters $[842,4,812]$ and weight enumerator
\[
1+20511120x^{812}+241497926880x^{840}+17230048080x^{841}+241497926880x^{842}.
\]
The dual code $\mathcal{C}(1,871)^{\perp}$ has parameters $[842,838,4]$. The codewords of weight $812$ in $\mathcal{C}(1,871)$ support a $3$-$(842,812,21897)$ design, and the codewords of weight $4$ in $\mathcal{C}(1,871)^{\perp}$ support a $3$-$(842,4,27)$ design. It is consistent with Theorem \ref{theorem1}.
\end{example}

\begin{example}
    Let $q=29$, $r=10$. Then $\mathcal{C}(1,871)$ has parameters $[842,4,812]$ and weight enumerator
\[
1+20511120x^{812}+241497926880x^{840}+17230048080x^{841}+241497926880x^{842}.
\]
The dual code $\mathcal{C}(1,871)^{\perp}$ has parameters $[842,838,4]$. The codewords of weight $812$ in $\mathcal{C}(1,871)$ support a $3$-$(842,812,21897)$ design, and the codewords of weight $4$ in $\mathcal{C}(1,871)^{\perp}$ support a $3$-$(842,4,27)$ design, which is also consistent with Theorem \ref{theorem1}.
\end{example}

\subsection{The Second Family of Constacyclic Codes Supporting 3-Designs}

Let   $r={\rm ord}(\lambda)$ satisfy the  following two conditions:

(1) $r \mid (q-1)$ and

(2) $\nu_{2}(r)= \nu_{2}(q-1)$.

\noindent Similarly as Subsection \ref{first class},  it is easy to check that the minimal polynomial of $\delta$ and $\delta^{q^{2}-q+1}$ are
$$
g_{1}(x)=  (x-\delta)(x-\delta^{q^{2}}) \quad {\rm and} \quad
g_{q^{2}-q+1}(x)=  (x-\delta^{q^{2}-q+1})(x-\delta^{q^{4}-q^{3}+q^{2}}),
$$ respectively.
Let $\mathcal{C}(1,q^{2}-q+1)$ denote the constacyclic code of length \( n = q^2 + 1 \) over \( \mathbb{F}_{q^2} \) with check polynomial \( g_{1}(x)g_{q^{2}-q+1}(x) \). In this subsection,  we study the parameters of the code \( \mathcal{C}(1, q^2 - q + 1) \) and its dual.

Define  $h(x)=x^{q-1}$ restricted on $U_{r(q^{2}+1)}$,  and the preimage set of $h(x)$ on $y_0$ is defined as $h^{-1}(y_0)=\{x\in U_{r(q^2+1)}\,\mid \,h(x)=y_0\}$, where $y_0\in U_{r(q^{2}+1)}$.
Then we have the following  result.

\begin{lem}\label{q-1 relation}
Let   notation be  the same as above, $y_{0} \in U_{q^{2}+1}$ and $T$ is defined as \eqref{T}. Then $h(x)$ is a surjective function from $U_{r(q^{2}+1)}$ to $U_{q^{2}+1}$. Moreover, if $h(x_0)=y_0$ for some $x_0\in U_{r(q^{2}+1)}$, then
  $h^{-1}(y_{0})=\{ \lambda^j x_0\mid 0\leq j \leq r-1 \}$ and $\mid h^{-1}(y_{0})\cap \lambda^{j}T \mid = 1$ for any $0\leq j \leq r-1$.
\end{lem}
\begin{proof}
  The proof is similar to that of Lemma \ref{q+1 relation}, and we thus omit it here.
\end{proof}
We now present the second important result of this paper.

\begin{thm}\label{theorem2} Let notation be the same as above and $q > 2$.
Then the code $\mathcal{C}(1,q^{2}-q+1)$ over $\mathbb{F}_{q^{2}}$ has parameters $[q^{2}+1,4,q^{2}-q]$ and weight enumerator
\[
1+(q^{5}-q)z^{q^{2}-q}+\frac{(q^{4}-1)(q-1)q^{3}}{2}(z^{q^{2}-1}+z^{q^{2}+1})+(q^{7}-q^{5}+q^{4}-q^{3}+q-1)z^{q^{2}}.
\]
Furthermore, the minimum weight codewords of the code $\mathcal{C}(1,q^{2}-q+1)$ support a $3$-$(q^{2}+1,q^{2}-q,(q^{2}-q-1)(q-2))$ design.
The code $\mathcal{C}(1,q^{2}-q+1)^{\perp}$ over $\mathbb{F}_{q^{2}}$ has parameters $[q^{2}+1,q^{2}-3,4]$ and the minimum weight codewords of $\mathcal{C}(1,q^{2}-q+1)^{\perp}$ support a $3$-$(q^{2}+1,4,q-2)$ design.
\end{thm}

\begin{proof} With the similar method of Theorem \ref{theorem1} and  using Lemma \ref{q-1 relation}, we then have the following two claims.

{\bf Claim 1.}  The possible nonzero weights of  $\mathcal{C}(1,q^{2}-q+1)$ take values from the set
$$\{q^2+1,q^2,q^2-1,q^2-q\}.$$

{\bf Claim 2.}  The dimension of the code  $\mathcal{C}(1,q^{2}-q+1)^{\perp}$ is $q^2-3$ and the minimum Hamming distance of $\mathcal{C}(1,q^{2}-q+1)^{\perp}$ is 4 (resp. 5) when $q>2$.

 Denote $w_{0}=q^{2}+1$, $w_{1}=q^{2}$, $w_{2}=q^{2}-1$ and $w_{3}=q^{2}-q$. Let $A_{w_{i}}$ denote the number of the codewords with weight $w_{i}$ in $\mathcal{C}(1,q^{2}-q+1)$, where $0\leq i\leq 3$. Since  the minimum Hamming distance of $\mathcal{C}(1,q^{2}-q+1)^{\perp}$ is $4$ by {\bf Claim 2}, we then have
\[
\begin{cases}
\sum_{i=0}^{3}A_{w_{i}}=q^{8}-1,\\
\sum_{i=0}^{3}w_{i}A_{w_{i}}=q^{6}(q^{4}-1),\\
\sum_{i=0}^{3}w_{i}^{2}A_{w_{i}}=q^{8}(q^{4}-1),\\
\sum_{i=0}^{3}w_{i}^{3}A_{w_{i}}=q^{16}+2q^{14}-2q^{12}-2q^{10}+q^{4}.
\end{cases}
\]
from the first four Pless power moments.  Solving this system of equations gives that
\[
\begin{cases}
A_{w_{3}}=q^{5}-q,\\
A_{w_{0}}=A_{w_{2}}=(q^{4}-1)(q-1)q^{3}/2,\\
A_{w_{1}}=q^{7}-q^{5}+q^{4}-q^{3}+q-1.
\end{cases}
\]
Using the Assmus-Mattson Theorem, one can immediately obtain that the minimum weight codewords in $\mathcal{C}(1,q^{2}-q+1)$ support a $3$-$(q^{2}+1,q^{2}-q,\eta)$ design. Note that the number of the supports of the minimum weight codewords in $\mathcal{C}(1,q^{2}-q+1)$ is
\[
b=\frac{A_{q^{2}-q}}{q^{2}-1}=\frac{q^{5}-q}{q^{2}-1}=q^{3}+q.
\]
Then by \eqref{t-design}, we have
\[
\eta=(q^{2}-q-1)(q-2).
\]
Further, one has that the number of codewords with weight $4$ in $\mathcal{C}(1,q^{2}-q+1)^{\perp}$  is
\[
A_{4}^{\perp}=\frac{q^{2}(q-2)(q^{2}+1)(q^{2}-1)^{2}}{24}.
\]
due to the fifth Pless power moment. Again by the Assmus-Mattson Theorem,  the codewords of weight $4$ in $\mathcal{C}(1,q^{2}-q+1)^{\perp}$ support a $3$-$(q^{2}+1,4,q-2)$ design.

This completes the proof.
\end{proof}

\begin{remark}
In \cite[Theorem 35]{wang2023infinite},  a family of negacyclic codes over $\mathbb{F}_{q^2}$  with the same parameters $[q^2+1, 4, q^2 - q]$ supporting a $3$-design, also was constructed. However, one should note that the condition $q \equiv 3 \pmod{4}$ is needed in their conclusion. Our Theorem \ref{theorem2} generalizes their result to a broader family of $\lambda$-constacyclic codes, where $\mathrm{ord}(\lambda) = r$ satisfies $r \mid (q-1)$ and $\nu_2(r) = \nu_2(q-1)$. This includes the negacyclic case ($r=2$) when $q \equiv 3 \pmod{4}$.
\end{remark}

\begin{remark}
Note that for the special case $r = 1$, $p$ can only be $2$. One can check that $\mathcal{C}(1,q^2+q+1)$ and $\mathcal{C}(1,q^2-q+1)$ have the same nonzeros, which means they are the same codes. In particular, $\mathcal{C}(1,q^2+q+1)$ is an MDS code over $\mathbb{F}_4$ with the parameter $[5,4,2]$ when $q = 2$.
\end{remark}

\begin{remark}
Since the two constacyclic codes $\mathcal{C}(1,q^2-q+1)$ and $\mathcal{C}(1,q^2+q+1)$ derived from our construction share the same parameters, it is natural to ask whether $\mathcal{C}(1,q^2-q+1)$ and $\mathcal{C}(1,q^2+q+1)$ are equivalent. Generally, determining the equivalence of two codes is challenging. While equivalent codes share the same parameters $[n,k,d]$ and weight distributions, the converse is not necessarily true. Moreover, due to computational limitations, it is also difficult to verify the equivalence of the two codes by Magma. On the other hand, since the parameter $\lambda$ of these two types of codes are not mutually inclusive when $r > 2$, we study their weight distributions separately, and the theoretical study on the equivalence between the two codes is indeed a valuable topic of future research.
\end{remark}

\begin{example}
  Let $q=16$, $r=5$. Then $\mathcal{C}(1,241)$ has parameters $[257,4,240]$ and weight enumerator
\[
1+1048560x^{240}+2013235200x^{255}+267448335x^{256}+2013235200x^{257}.
\]
The dual code $\mathcal{C}(1,241)^{\perp}$ has parameters $[257,253,4]$. The codewords of weight $240$ in $\mathcal{C}(1,241)$ support a $3$-$(257,240,3346)$ design, and the codewords of weight $4$ in $\mathcal{C}(1,241)^{\perp}$ support a $3$-$(257,4,14)$ design, which satisfy Theorem \ref{theorem2}.
\end{example}

\begin{example}
  Let $q=13$, $r=4$. Then $\mathcal{C}(1,157)$ has parameters $[170,4,156]$ and weight enumerator
\[
1+371280x^{156}+376477920x^{168}+62403600x^{169}+376477920x^{170}.
\]
The dual code $\mathcal{C}(1,157)^{\perp}$ has parameters $[170,166,4]$. The codewords of weight $156$ in $\mathcal{C}(1,157)$ support a $3$-$(170,156,1705)$ design, and the codewords of weight $4$ in $\mathcal{C}(1,157)^{\perp}$ support a $3$-$(170,4,11)$ design. It is consistent with Theorem \ref{theorem2}.
\end{example}

\begin{example}
    Let $q=25$, $r=8$. Then $\mathcal{C}(1,601)$ has parameters $[626,4,600]$ and weight enumerator
\[
1+9765600x^{600}+73242000000x^{624}+6094125024x^{625}+73242000000x^{626}.
\]
The dual code $\mathcal{C}(1,601)^{\perp}$ has parameters $[626,622,4]$. The codewords of weight $600$ in $\mathcal{C}(1,601)$ support a $3$-$(626,600,13777)$ design, and the codewords of weight $4$ in $\mathcal{C}(1,601)^{\perp}$ support a $3$-$(626,4,23)$ design, which is consistent with Theorem \ref{theorem2}.
\end{example}

\section{Applications of the constacyclic codes}

In this section, we will further explore the  applications of the constacyclic codes proposed in Section \ref{main:result}, focusing on the following three aspects.  We firstly investigate the structure of their subfield subcodes, revealing their equivalence to classical ovoid codes. By utilizing the intersection properties of these codes and their duals, we construct entanglement-assisted quantum error-correcting codes (EAQECCs), demonstrating their potential in the field of quantum information. In addition, we analyze the  performance of their dual codes in the construction of locally recoverable codes (LRCs), proving that they can achieve certain bounds.

\subsection{Subfield subcodes of constacyclic codes}

Let $\mathcal{C}$ be an $[n,k]$ code over $\mathbb{F}_{q^m}$, where $q$ is a prime power and $m$ is a positive integer. The subfield subcode of $\mathcal{C}$ over $\mathbb{F}_{q}$, denoted by $\mathcal{C}|_{\mathbb{F}_{q}}$, is defined by
\[
\mathcal{C}|_{\mathbb{F}_{q}} = \mathcal{C} \cap \mathbb{F}_{q}^{n}.
\]

A fundamental result due to Delsarte establishes a profound connection between a subfield subcode and the trace code of its dual.

\begin{thm}{\rm (\cite[Delsarte Theorem]{Delsarte1998})}\label{subcode}
Let $\mathcal{C}$ be a linear code of length $n$ over $\mathbb{F}_{q^m}$. Then
\[
\mathcal{C}|_{\mathbb{F}_{q}} = \left( \mathrm{Tr}_{q^m/q}(\mathcal{C}^\perp) \right)^\perp.
\]
\end{thm}

 An ovoid in the projective space \(\mathrm{PG}(3, \mathbb{F}_{q})\) is a set of \( q^2 + 1 \) points such that no three of them are collinear (i.e., on the same line). An ovoid code is a linear code over \(\mathbb{F}_{q}\) with parameters \([q^2 + 1,4,q^2 - q]\). There exists a one-to-one correspondence between ovoids and ovoid codes. The following theorem summarizes key properties of these codes.

\begin{thm}{\rm (\cite{ding2022designs})}\label{ovoid}
Let \( q > 2 \) be a prime power. Then

{\rm (1)} Every ovoid code \(\mathcal{C}\) over \(\mathbb{F}_{q}\) must have parameters \([q^2 + 1,4,q^2 - q]\) and weight enumerator
\[
1 + (q^2 - q)(q^2 + 1)z^{q^2 - q} + (q - 1)(q^2 + 1)z^{q^2}.
\]
When \( q \) is odd, all ovoid codes over \(\mathbb{F}_{q}\) are monomially equivalent. When \( q \) is even, the elliptic quadric code and the Tits ovoid code over \(\mathbb{F}_{q}\) are not monomially equivalent.

{\rm (2)} The dual of every ovoid code \(\mathcal{C}\) over \(\mathbb{F}_{q}\) must have parameters \([q^2 + 1, q^2 - 3,4]\).

 {\rm (3)} The minimum weight codewords in an ovoid code over \(\mathbb{F}_{q}\) support a \(3\)-\((q^2 + 1, q^2 - q, (q-2)(q^2 - q-1))\) design and the complementary design of this design is a Steiner system \(S(3,1+q,1+q^2)\). When \( q \) is odd, all these \(3\)-\((q^2 + 1, q^2 - q, (q-2)(q^2 - q-1))\) designs are isomorphic. The minimum weight codewords in the dual of an ovoid code over \(\mathbb{F}_{q}\) support a \(3\)-\((q^2 + 1, 4, q-2)\) design.
\end{thm}
We now present the main results of this subsection, which characterizes the subfield subcodes of the constacyclic codes constructed in Section 3.

\begin{thm}\label{subcode2}
The subfield subcode $\mathcal{C}\left(1,q^{2}-q+1\right)|_{\mathbb{F}_{q}}$ has parameters $[q^{2}+1,4,q^{2}-q]$ and weight enumerator
\[
1+(q^{2}-q)(q^{2}+1)z^{q^{2}-q}+(q-1)(q^{2}+1)z^{q^{2}}.
\]
The dual code of $\mathcal{C}\left(1,q^{2}-q+1\right)|_{\mathbb{F}_{q}}$ has parameters $[q^{2}+1,q^{2}-3,4]$.
\end{thm}

\begin{proof}
From Theorem \ref{subcode},  one has that
\[
\mathcal{C}\left(1,q^{2}-q+1\right)|_{\mathbb{F}_{q}} = \left( \mathrm{Tr}_{q^{2}/q} \left( \mathcal{C} \left(1,q^{2}-q+1\right)^{\perp} \right) \right)^{\perp}.
\]
We first consider the generator polynomial of $\mathrm{Tr}_{q^{2}/q} \left( \mathcal{C} \left(1,q^{2}-q+1\right)^{\perp} \right)$.
By the definition of $\mathcal{C}\left(1,q^{2}-q+1\right)$, we know that $\delta$, $\delta^{q}$, $\delta^{q^{2}}$ and $\delta^{q^{2}-q+1}$ are all nonzeros of $\mathcal{C}\left(1,q^{2}-q+1\right)$. Then from Lemma \ref{generator and check poly}, all zeros of $\mathcal{C}\left(1,q^{2}-q+1\right)^{\perp}$ are $\delta^{rn-1}$, $\delta^{rn-q}$, $\delta^{rn-q^2}$ and $\delta^{rn-q^{2}+q-1}$.

Let the $q$-cyclotomic coset of $s$ modulo $rn$ be denoted by $C_{s}^{(q,rn)}$, then
\[
C_{-1}^{(q,rn)}=\left\{rn-1,\,rn-q,\,rn-q^2,\,rn-q^2+q-1\right\}.
\]
Then it follows from Lemma \ref{trace exp} that $\mathrm{Tr}_{q^{2}/q}\left(\mathcal{C}\left(1,q^{2}-q+1\right)^{\perp}\right)$ is a constacyclic code over $\mathbb{F}_{q}$ with generator polynomial
\[
\left( x - \delta^{rn-1} \right) \left( x - \delta^{rn-q} \right) \left( x - \delta^{rn-q^2} \right) \left( x - \delta^{rn-q^2+q-1} \right).
\]
Hence, the dimension of $\mathcal{C}\left(1,q^{2}-q+1\right)|_{\mathbb{F}_{q}}$ is 4.

From Theorem \ref{theorem2}, $\mathcal{C}\left(1,q^{2}-q+1\right)$ has minimum distance $q^{2}-q$. Then it follows from the definition of subfield subcodes that the minimum distance $d \left( \mathcal{C}\left(1,q^{2}-q+1\right)|_{\mathbb{F}_{q}} \right) \geq q^{2}-q$. By the Griesmer bound \cite{Griesmer1960}, we know that $d \left( \mathcal{C}\left(1,q^{2}-q+1\right)|_{\mathbb{F}_{q}} \right) \leq q^{2}-q$. Hence, we deduce that $\mathcal{C}\left(1,q^{2}-q+1\right)|_{\mathbb{F}_{q}}$ has parameters $[q^{2}+1,4,q^{2}-q]$. This means that $\mathcal{C}\left(1,q^{2}-q+1\right)|_{\mathbb{F}_{q}}$ is an ovoid code. The weight distribution of $\mathcal{C}\left(1,q^{2}-q+1\right)|_{\mathbb{F}_{q}}$ and the parameters of the dual code then follow from Theorem \ref{ovoid}. This completes the proof.
\end{proof}

\begin{example}
Let $q=13$, $r=4$. Magma experiments show that $\mathcal{C}(1,157)|_{\mathbb{F}_{13}}$ has parameters $[170,4,156]$ and weight enumerator
\[
1+26520x^{156}+2040x^{169}.
\]
The dual code of $\mathcal{C}(1,157)|_{\mathbb{F}_{13}}$ has parameters $[170,166,4]$. It is consistent with Theorem \ref{subcode2}.
\end{example}

The following theorem provides the conditions for distinguishing subfield subcodes when $\lambda$ is not in the subfield.

\begin{thm}\label{b2017}\cite{Blackford2017}
Let $\mathcal{C}$ be a $\lambda$-constacyclic code over $\mathbb{F}_{q^{m}}$ of length $n$, with $\lambda \in \mathbb{F}_{q^{m}}^{\times}$ of order $r$. 
Let $k$ be the multiplicative order of $q$ modulo $r$, so that $\mathbb{F}_{q^{k}}$ is the smallest subfield of $\mathbb{F}_{q^{m}}$ containing $\lambda$, and write $m = kl$. 
Let $g(x)$ be the generator polynomial of $\mathcal{C}$, and let $T$ be the set of index of zeros of generator polynomial $g(x)$. Define
\[
T_{2} = T \cup q^{k}T \cup \cdots \cup q^{k(l-1)}T.
\]
Then the subfield code $\mathcal{C}|_{\mathbb{F}_{q}}=0$ if and only if $|T_{2}| \geq \frac{n}{k}$.
\end{thm}

\begin{thm}\label{subcode1}
Let $\mathcal{C}(1,q^{2}+q+1)$ be the constacyclic code defined in Section 3.1 and $r\neq 1.$ Then it holds that $$\mathcal{C}(1,q^{2}+q+1)|_{\mathbb{F}_{q}} = \{0\},$$ i.e., its subfield subcode over $\mathbb{F}_{q}$ is trivial.
\end{thm}

\begin{proof}
Note that for the special case $r=1$, as we point out in Remark \ref{the same}, the $\mathcal{C}(1,q^{2}+q+1)|_{\mathbb{F}_{q}}$ is just give in Theorem \ref{subcode2}.

When $r=2$, by Theorem \ref{subcode} we know that
\[
\mathcal{C}(1,q^{2}+q+1)|_{\mathbb{F}_{q}} = \left( \mathrm{Tr}_{q^{2}/q}\left( \mathcal{C} (1,q^{2}+q+1)^{\perp} \right) \right)^{\perp}.
\]
In order to obtain the desired result, we will show all the possible zeros of $\mathrm{Tr}_{q^{2}/q}(\mathcal{C}(1,q^{2}+q+1)^{\perp})$ are exactly nonzeros.

By the definition of $\mathcal{C}(1,q^{2}+q+1)$, it was shown earlier that all the nonzeros of $\mathcal{C}(1,q^{2}+q+1)$ are
$
\delta, \delta^{q^{2}}, \delta^{q^{2}+q+1}, \delta^{-q}.
$
 Then by Lemma \ref{generator and check poly}, all zeros of $\mathcal{C}(1,q^{2}+q+1)^{\perp}$ are $\delta^{rn-1}$, $\delta^{rn-q^2}$, $\delta^{rn-q^2-q-1}$ and $\delta^{q}$. According to Lemma \ref{trace exp}, the nonzeros of $\mathcal{C}(1,q^{2}+q+1)^{\perp}$ are the nonzeros of $\mathrm{Tr}_{q^{2}/q}(\mathcal{C}(1,q^{2}+q+1)^{\perp})$. Then the possible zeros of $\mathrm{Tr}_{q^{2}/q}(\mathcal{C}(1,q^{2}+q+1)^{\perp})$ are $\delta^{rn-1}$, $\delta^{rn-q^2}$, $\delta^{rn-q^2-q-1}$ and $\delta^{q}$.
Let the $q$-cyclotomic coset of $s$ modulo $rn$ be denoted by $C_{s}^{(q,rn)}$, then
\[
C_{-1}^{(q,rn)} = \{-1, -q, -q^{2}, q^{2}+q+1\}
\]
and
\[
C_{q}^{(q,rn)} = \{q, q^{2}, -q^2-q-1, 1\}.
\]
If one element of the coset is a zero, then all the elements of the coset must be also zero. However, it can be seen that every possible zero has a conjugate element as nonzero when $r\neq 1$.  Thus we can show that $\delta^{rn-1}$, $\delta^{rn-q^2}$, $\delta^{rn-q^2-q-1}$ and $\delta^{q}$ are nonzeros of $\mathrm{Tr}_{q^{2}/q}(\mathcal{C}(1,q^{2}+q+1)^{\perp})$.

When $r > 2$, by Theorem~\ref{b2017},  $\mathcal{C}(1,q^{2}+q+1)|_{\mathbb{F}_{q}} = \{0\}$ if and only if $|T_{2}| \geq \frac{n}{2}$. Note that $|T_{2}| \ge |T|$, where $|T|=n-4$ and $|T| \ge \frac{n}{2} $ always established when $q >2$.
 This completes the proof.
\end{proof}

\subsection{EAQECCs from constacyclic codes}
Let $[[n,k,d;c]]_{q}$ denote a $q$-ary EAQECC which encodes $k$ logical qubits into $n$ physical qubits by means of $c$ copies of maximally entangled states~\cite{Wilde2008}. It is known that $q$-ary EAQECCs can be obtained through two $\mathbb{F}_{q}$-linear codes $\mathcal{C}_{1}$ and $\mathcal{C}_{2}$.

\begin{lem}{\rm (\cite{LiuChen2025Intersection})}
\label{eaqecc}
Assume that $\mathcal{C}_{1}$ and $\mathcal{C}_{2}$ are two $q$-ary linear codes with parameters $[n,k_{1},d_{1}]$ and $[n,k_{2},d_{2}]$ respectively. Then there exists an $[[n,k_{1}+k_{2}-n+c,\min\{d_{1},d_{2}\};c]]_{q}$ EAQECC, where $c=n-k_{1}-\dim(\mathcal{C}_{1}^{\perp} \bigcap \mathcal{C}_{2})$.
\end{lem}

If $c = n - k$, then it is a maximal entanglement EAQECC. The performance of an EAQECC is measured using the rate $\frac{k}{n}$ and net rate $\frac{k-c}{n}$. The net rate of an EAQECC can be positive, negative, or zero. EAQECCs with positive net rates are used as catalytic codes in quantum computing \cite{Brun2006}. An EAQECC is a standard stabilizer code if $c = 0$.

To determine the parameter $c$, we utilize the following result on the intersection of constacyclic codes found in \cite{Hossain2023}.

\begin{lem}{\rm (\cite{Hossain2023})}
\label{the value of c}
Let $\mathcal{C}_{i} = \langle g_{i} \rangle$ be a $\lambda_{i}$-constacyclic code with parameter $[n, k_{i}]$ over $\mathbb{F}_{q}$, where  $i = 1, 2$, and $\ell = \dim(\mathcal{C}_{1} \cap\mathcal{C}_{2})$. Let $h_{i} = (x^{n} - \lambda_{i})/g_{i}$ be a parity polynomial corresponding to $g_{i}$. If $\lambda_{1} \neq \lambda_{2}$, then
\[
\ell =
\begin{cases}
0, & \text{for } k_{1} + k_{2} \leq n \\
\deg h_{2} - \deg g_{1}, & \text{for } k_{1} + k_{2} > n.
\end{cases}
\]
\end{lem}

We now present the main result of this subsection, which gives two families of EAQECCs derived from the constacyclic codes   in Section \ref{main:result}.

Let $\lambda(\mathcal{C})$ denote the shift constant of $\lambda$-constacyclic code $\mathcal{C}$.
Define $$T=\{\mathcal{C} \,|\,\mathcal{C} = \mathcal{C}\left(1,q^{2}+q+1\right) \;\mathrm{or} \; \mathcal{C}\left(1,q^{2}-q+1\right) \;\mathrm{constructed \;in\; Section \; \ref{main:result}} \}$$ and $$T^{\perp}=\{\mathcal{C}=\mathcal{C}_{0}^{\perp}\,|\,\mathcal{C}_{0} \in T\}.$$
\begin{thm}\label{reult:eaqecc}
With the notation above, we have

    {\rm (1)} If $\mathcal{C}_{1},\mathcal{C}_{2} \in T$ and $\lambda(\mathcal{C}_{1})\lambda(\mathcal{C}_{2}) \neq 1$, then there exists an $[[q^2+1, 4, q^2-q; q^2-3]]_{q^2}$ EAQECC, which is a maximal entanglement EAQECC with a negative net rate.

    {\rm (2)} If $\mathcal{C}_{1},\mathcal{C}_{2} \in T^{\perp}$ and $\lambda(\mathcal{C}_{1})\lambda(\mathcal{C}_{2}) \neq 1$, then there exists an $[[q^2+1, q^2-3, 4; 4]]_{q^2}$ EAQECC, which is a maximal entanglement EAQECC with a high positive net rate.

\end{thm}
\begin{proof}
   {\rm (1)}  According to Lemma \ref{eaqecc}, there exists an $[[n, k_1 + k_2 - n + c, min\{d_1, d_2\}; c]]_{q^2}$ EAQECC where $\mathcal{C}_{1},\mathcal{C}_{2} \in T$ and $c=n-k_{1}-\dim(\mathcal{C}_{1}^{\perp} \bigcap \mathcal{C}_{2})$.
           Hence, in order to obtain the desired result, we only need to show the exact value of c.
           Because of $\lambda(\mathcal{C}_{1})\lambda(\mathcal{C}_{2}) \neq 1$, we have $\lambda(\mathcal{C}^{\perp}_{1}) \neq \lambda(\mathcal{C}_{2})$. Then by Lemma \ref{the value of c}, we have dim$(\mathcal{C}^{\perp}_{1} \cap \mathcal{C}_{2})=0$,
           which means $c=n-k_{1}=q^2-3$. Thus we can give the parameters of the EAQECC as $[[q^2+1, 4, q^2-q; q^2-3]]_{q^2}$, which means this is a maximal entanglement EAQECC with a negative net rate.

   {\rm (2)}   Note that $\mathcal{C}_{1},\mathcal{C}_{2} \in T^{\perp}$ and $\lambda(\mathcal{C}_{1})\lambda(\mathcal{C}_{2}) \neq 1$. Then with the similar method of (1), we have $c=4$ and the parameters of the EAQECC are $[[q^2+1, q^2-3, 4; 4]]_{q^2}$, which is a maximal entanglement EAQECC. Since the net rate is $\frac{q^2-7}{q^2+1}$, it   approaches  to 1 when $q$ is big enough.

This completes  the proof.
\end{proof}

\begin{remark} To the best of our knowledge, we are the first to offer  EAQECCs which have the parameters as  presented  in Theorem \ref{reult:eaqecc}. The first code of  Theorem \ref{reult:eaqecc} attains maximal entanglement with a high error-correction capability (large $d$), while the second achieves a near-optimal net rate close to 1. Both families are valuable in quantum communication and computation, offering either enhanced error suppression or high resource efficiency.
\end{remark}

\subsection{Locally recoverable codes from constacyclic codes}

The $i$-th code symbol of a linear code $\mathcal{C}$ is said to have locality $r_{l}$, if $x_i$ in each codeword $x \in \mathcal{C}$ can be recovered by $r_{l}$ other code symbols (i.e., $x_i$ is a function of some other $r_{l}$ symbols $x_{i_1}, x_{i_2}, \ldots, x_{i_r}$). A linear code is said to have locality $r_{l}$ if every code symbol has locality at most $r_{l}$. An $(n, k, d, q; r_{l})$-LRC (for short) is a linear code of  length $n$, dimension $k$, minimum distance $d$ and locality $r_{l}$ over $\mathbb{F}_q$ \cite{ref14}.

The following two bounds gave tradeoffs of LRC between the code length, dimension, and the minimum distance which can be characterized by some inequalities.

\begin{lem}{\rm (\cite[Singleton-like bound]{ref5})}\label{lrcbound1}
For an $(n, k, d, q; r_{l})$-LRC, we have
\begin{equation}\label{bound1}
d \leq n - k - \left\lceil \frac{k}{r_{l}} \right\rceil + 2.
\end{equation}
\end{lem}

\begin{lem}{\rm (\cite[Cadambe-Mazumdar bound]{ref1})}\label{lrcbound2}
For an $(n, k, d, q; r_{l})$-LRC, we have
\begin{equation}\label{bound2}
k \leq \min_{t \in \mathbb{Z}_+} \left\{ tr + k^{(q)}_{\mathrm{opt}}(n - t(r_{l} + 1), d) \right\},
\end{equation}
where $k^{(q)}_{\mathrm{opt}}(n, d)$ denotes the largest possible dimension of a linear code with code length $n$, minimum distance $d$ and alphabet size $q$, where $\mathbb{Z}_+$ is the set of all nonnegative integers.
\end{lem}

For any given $(n, k, d, q; r_{l})$-LRC, we say the code is dimension-optimal if the parameters of the code meet the Cadambe-Mazumdar bound in \eqref{bound2}. Similarly, an $(n, k, d, q; r_{l})$-LRC is said to be distance-optimal if the distance $d$ achieves the Singleton-like bound in \eqref{bound1}.

The following is a simple result on the locality of constacyclic codes found in \cite{zhao2020} , and one can find a more generalized proof in \cite{ref15}.
\begin{lem}{\rm (\cite{zhao2020})}\label{LRC}
Let $\mathcal{C}$ be a constacyclic code with dual distance $d^{\perp} > 2$. Then
the minimum locality of $\mathcal{C}$ is $d^{\perp} - 1$.
\end{lem}
We now show that the duals of our constacyclic codes yield optimal LRCs.
\begin{thm}\label{LRC thm}
Let $\mathcal{C}$ be the constacyclic code with the form $\mathcal{C}(1,q^2+q+1)$ or $\mathcal{C}(1,q^2-q+1)$ over $\mathbb{F}_{q^2}$. Then the dual code $\mathcal{C}^{\perp}$ is a distance-optimal and dimension-optimal LRC.
\end{thm}
\begin{proof}
  By Lemma \ref{LRC}, the locality $r_{l}$ of $\mathcal{C}^{\perp}$ is $q^2-q-1$.
  According to Lemma \ref{lrcbound1}, one has that the Singleton-like bound of $\mathcal{C}^{\perp}$ is
  $$(q^2+1)-(q^2-3)-\left\lceil\frac{q^2-3}{q^2-q-1}\right\rceil+2=4,$$
  which is equal to the minimum distance $d^{\perp}=4$ of $\mathcal{C}^{\perp}$. This means that $\mathcal{C}^{\perp}$ is distance-optimal.
  Since the length of a code must be possitive, $t$ in Cadambe-Mazumdar bound can only be 1, then $k^{(q)}_{\mathrm{opt}}(n-t(r+1), d)$ is $q-2$ by the singleton bound.
  Then according to Lemma \ref{lrcbound2}, we have that Cadambe-Mazumdar bound of $\mathcal{C}^{\perp}$ is
  $$(q^2-q-1)+(q-2)=q^2-3,$$
   which induces that $\mathcal{C}^{\perp}$ is also dimension-optimal.
\end{proof}

\section{Conclusions}

This paper has explored the properties of two infinite families of $\lambda$-constacyclic codes over $\mathbb{F}_{q^2}$. The main contributions of our investigation can be summarized as follows.

(1) We have established two distinct families of $\lambda$-constacyclic codes of length $n = q^2 + 1$, characterized by specific conditions on the order of $\lambda$:

   {\rm (i)} Codes $\mathcal{C}(1,q^2+q+1)$ with $r \mid (q+1)$ and $\nu_2(r) = \nu_2(q+1)$.

   {\rm (ii)} Codes $\mathcal{C}(1,q^2-q+1)$ with $r \mid (q-1)$ and $\nu_2(r) = \nu_2(q-1)$.

\noindent These codes have parameters $[q^2+1, 4, q^2-q]$ and support $3$-$(q^2+1, q^2-q, (q^2-q-1)(q-2))$ designs. Their dual codes, with parameters $[q^2+1, q^2-3, 4]$, similarly support $3$-$(q^2+1, 4, q-2)$ designs.

(2) The intersection properties of these codes have enabled the construction of several families of entanglement-assisted quantum error-correcting codes. Additionally, the dual codes exhibited optimal characteristics of locally recoverable codes, which meets the Singleton-like bound and Cadambe-Mazumdar bound.

(3) We provided an alternative proof of the solutions on a class of  equations compared to \cite{wang2023infinite} and gave  a positive answer to the conjecture proposed by \cite{conj}. In additions, our results totally  extended the results on negacyclic codes  presented in  \cite{wang2023infinite}.

The results suggest several interesting directions for future research, including further investigation of constacyclic codes supporting higher-order designs, exploring their applications in coding theory and other fields.

\section*{Acknowledgment}
This work was supported by  the National Natural Science Foundation of China (Nos. 12501738, 12201356, 62032009, 12471492),   the National Key Research and Development Program of China (No. 2021YFA1000600), the Innovation Group Project of the Natural Science Foundation of Hubei Province of China (No. 2023AFA021).

\end{document}